\documentclass[a4paper,reqno,10pt]{amsart}
\usepackage[latin1]{inputenc}
\usepackage[T1]{fontenc}

\usepackage{amssymb,amsmath}
\usepackage{setspace}
\usepackage[normalem]{ulem}
\usepackage{color}
\usepackage{times}
 \usepackage{natbib}
\usepackage[english]{babel}

\newtheorem{proposition}{Proposition}

\begin{document}

\title[Are weighted games sufficiently good for binary voting?]{Are weighted games sufficiently good for binary voting?}

\author[Are weighted games sufficiently good for binary voting?]{Sascha Kurz\\ Department of Mathematics, University of Bayreuth, Germany.\\ Tel: +49 921 557353. E-mail: sascha.kurz@uni-bayreuth.de}  



\maketitle

\setstretch{1.2}


\noindent
\textsc{Abstract:}
Binary {\lq\lq}yes{\rq\rq}-{\lq\lq}no{\rq\rq} decisions in a legislative committee or a shareholder meeting are commonly 
modeled as a weighted game. However, there are noteworthy exceptions. E.g., the voting rules of the European Council according 
to the Treaty of Lisbon use a more complicated construction. Here we want to study the question if we lose much from a 
practical point of view, if we restrict ourselves to weighted games. To this end, we invoke power indices that measure the 
influence of a member in binary decision committees. More precisely, we compare the achievable power distributions of weighted 
games with those from a reasonable superset of weighted games. It turns out that the deviation is relatively small.

\medskip


\medskip



\section{Introduction}
\label{sec:intro}
\noindent
Consider a family, consisting of mother Ann, father Bob, and the two kids Cathrin and Dave, deciding on their joint weekend 
activities by binary voting. In a \emph{weighted game} each voter $i$ has a non-negative weight $w_i$ and a proposal is accepted if the sum of the weights 
of its supporters meets or exceeds a positive quota $q$. As an abbreviation we write $\left[q;w_1,\dots,w_n\right]$ for a weighted 
game with $n$ voters. The example $[3;3,2,1,1]$ (where we number in alphabetical order) might model a slightly parents-biased, especially mother-biased, decision rule. 
Another voting rule might be that either both parents or both kids have to agree. It can be shown that no representation as a 
weighted game exists. Since all family members have equal opportunities to influence the final decision, all reasonable 
measures of voting power assign equal power to all members. This is also true for other weighted games such as $[2;1,1,1,1]$ or $[3;1,1,1,1]$ (but 
not for $[3;3,2,1,1]$).  
If we only care about the resulting power distribution we can also choose 
a weighted game in our situation. Even more practically, we may accept a weighted game as a plausible replacement of the original 
voting rule if the corresponding power distribution does not differ too much. Here we want to study the question how large this 
difference can be in the worst case. Reasons for choosing only weighted games as decision rules is that they are easy to understand and to 
implement as well as quite common in practice. However, the more important point is the issue of representation complexity. For a weighted 
game with $n$ voters we just need $n+1$ numbers for the description while exponentially many coalitions have to be listed for binary decision rules 
in the worst case. 

A related problem is the so-called \emph{inverse power index problem}, where one wants to determine the game whose power distribution 
is closest to a predefined target power distribution. For more details see e.g.\ \cite{de2017inverse} and the references cited therein.  
\cite{alon2010inverse} have shown that some target power distributions, where most players have negligible or even zero 
power, like e.g.\ $(0.75,0.25,0,\dots,0)$, cannot be approximated too closely by the power distribution of any game.\footnote{More precisely, 
\cite{alon2010inverse} show such a result for the Banzhaf index. Results for other power indices have been obtained by \cite{kurz2016inverse}.} 
Our setting differs as follows. Instead of all non-negative vectors summing to one, we only consider the power distributions attained 
by a superset of weighted games as possible target power distributions and ask to what extent they can be approximated by the power 
distribution of a weighted game.   
 
\section{Preliminaries} 
\label{sec_preliminaries}
\noindent
By $N=\{1,\dots,n\}$ we denote the set of voters. A \emph{simple game} is a surjective and monotone mapping $v\colon 2^N\to\{0,1\}$ 
from the set of subsets of $N$ into a binary output $\{0,1\}$. \emph{Monotone} means $v(S)\le v(T)$ for all $\emptyset\subseteq S\subseteq T 
\subseteq N$. A simple game $v$ is weighted if there exist weights $w_1,\dots,w_n\in\mathbb{R}_{\ge 0}$ and a quota 
$q\in\mathbb{R}_{>0}$ such that $v(S)=1$ iff $w(S):=\sum_{i\in S} w_i\ge q$. As stated in the introduction, we abbreviate a weighted game 
by $\left[q;w_1,\dots,w_n\right]$. Two voters $i$ and $j$ are called \emph{symmetric}, in a given simple game $v$, if $v(S\cup\{i\})=v(S\cup\{j\})$ 
for all $\emptyset\subseteq S\subseteq N\backslash\{i,j\}$. Voter $i\in N$ is a \emph{null voter} if $v(S)=v(S\cup\{i\})$ for all 
$\emptyset\subseteq S\subseteq N\backslash\{i\}$. 

Given two simple games $v$ and $v'$ we define their \emph{intersection} or \emph{conjunction} $v\wedge v'$ via 
$(v\wedge v')(S)=\min\left\{v(S),v'(S)\right\}$ for all $S\subseteq N$. Similarly, the \emph{union} or \emph{disjunction} is given by $(v\vee v')(S)=\max\left\{v(S),v'(S)\right\}$ 
for all $S\subseteq N$. The non-weighted decision rule from the introduction can be written as $[2;2,0,1,1]\wedge[2;0,2,1,1]$ or $[2;1,1,0,0]\vee[2;0,0,1,1]$. 
It is well known, see e.g.\ \cite{taylor1999simple}, that every simple game can be written as the intersection (or union) of a finite list of weighted games. Also combinations 
of $\wedge$ and $\vee$ are used in practice. 

An example is given by the voting system of the European Council according to the Treaty of Lisbon. For $n=27$ 
(or $n=28$) countries the voting system can be written as $v=\left([0.55n;1,\dots,1]\wedge\left[0.65;p_1,\dots,p_{n}\right]\right)\vee [n-3;1,\dots,1]$, where 
$p_i$ denotes the relative population of country~$i$. As remarked by \cite{kirsch2011invariably}, dropping the union with $[n-3;1,\dots,1]$ has almost 
no impact on the characteristic function $v$ or corresponding power distributions. Consisting of a Boolean combination, i.e., $\wedge$'s and $\vee$'s, of 
three weighted games the stated representation of the voting system of the European Council (according to the Treaty of Lisbon) is relatively compact. For a
general simple game for $n$ voters an exponential number of weighted games can be necessary in the worst case, see \cite{faliszewski2009boolean}. Writing down 
the characteristic function $v$ explicitly also has exponential complexity, while a weighted game can be written by listing $n$ integer weights and a quota. 
Framed differently, the number of distinct simple games is many orders of magnitudes larger than the number of weighted games. 

As a class of binary voting systems  between simple games and weighted games we consider \emph{complete simple games}, see  \cite{carreras1996complete}. 
They are based on \emph{Isbell's desirability relation}, see \cite{isbell1956class}, where we write $i\succeq j$ if $v(S\cup\{i\})\ge v(S\cup\{j\})$ for all $S\subseteq N\backslash\{i,j\}$ for 
two voters $i,j\in N$. A simple game $v$ is called complete if this relation is complete, i.e., if for all $i,j\in N$ we have $i\succeq j$ or $j\succeq i$. Two players 
$i,j\in N$ are symmetric iff $i\succeq j$ and $j\succeq i$. The relation $\succeq$ induces an ordering of the players, which is satisfied in many practical applications. 
E.g.\ the voting systems of the European Council (according to the Treaty of Lisbon and also those before) are complete simple games. Here we use the standard assumption 
$1\succeq 2\succeq \dots\succeq n$ and note that Shapley-Shubik index $\operatorname{SSI}(v)$ and the Penrose-Banzhaf index $\operatorname{PBI}(v)$, 
see the definitions below, are non-increasing vectors for every complete simple game $v$.  
In order to uniquely characterize a complete simple game $v$ we can list all subsets $S\subseteq N$ such that $v(S)=1$ and for every $i\in S$, $j\notin S$ with $i<j$ (using 
the usual ordering in $\mathbb{N}$)  
we have $v(S\backslash\{i\}\cup\{j\})=0$. For our example $[3;3,2,1,1]$ those (so-called {\lq}shift-minimal winning{\rq}) subsets are given by $\{1\}$ and $\{2,4\}$. 
In our example $[2;2,0,1,1]\wedge[2;0,2,1,1]$ the 
voters $1$ and $2$ as well as voters $3$ and $4$ are symmetric. For all other pairs of different voters we neither have $i\succeq j$ nor $j\succeq i$, i.e., the game 
is not complete. 

A power index $p$ is a mapping from the set of simple (or weighted) games on $n$ voters into $\mathbb{R}^n$. By $p_i(v)$ we denote the $i$th component of 
$p(v)$, i.e., the power of voter~$i$. Here we consider two of the most commonly used power indices, i.e., the 
\emph{Shapley-Shubik index}, see \cite{shapley1954method}, 
$$
  \operatorname{SSI}_i(v)=\sum_{S\subseteq N\backslash\{i\}} \frac{|S|!\cdot(n-|S|-1)!}{n!}\cdot\left(v(S\cup\{i\})-v(S)\right)
$$
and the \emph{Penrose-Banzhaf index}, see \cite{penrose1946elementary,banzhaf1964weighted},
$$
  \operatorname{PBI}_i(v)=\frac{\sum_{S\subseteq N\backslash\{i\}}\left(v(S\cup\{i\})-v(S)\right)}{\sum_{j\in N} \sum_{S\subseteq N\backslash\{j\}}\left(v(S\cup\{j\})-v(S)\right)}.  
$$
For our first example $v=[3;3,2,1,1]$ we have $$\operatorname{SSI}(v)=\tfrac{1}{12}\cdot (7,3,1,1)\approx(0.5833,0.25,0.0833,0.0833)$$ and $$\operatorname{PBI}(v)=
\tfrac{1}{10}\cdot (5,3,1,1)=(0.5,0.3,0.1,0.1).$$ As a measure for the distance between two different power distributions $x,y\in\mathbb{R}^i$ we use the \emph{Manhattan 
distance} $d_1(x,y)=\sum_{i=1}^n \left|x_i-y_i\right|$ and the \emph{Chebyshev distance} $d_{\infty}(x,y)=\max_{1\le i\le n} \left| x_i-y_i\right|$. For the above two 
power distributions the Manhattan distance is $\tfrac{1}{6}\approx 0.1667$ and the Chebyshev distance is $\tfrac{1}{12}\approx 0.0833$.

\section{Results}
\label{sec_main}
\noindent
In this section we want to present some numerical results obtained by exhaustive enumeration. In order to shed some light on the different classes of binary voting rules we mention that 
up to $3$ voters all simple games are weighted. For $3$ voters the $8$ different weighted games are given by $[1;1,0,0]$, $[1;1,1,0]$, $[2;1,1,0]$, $[1;1,1,1]$, $[2;1,1,1]$, 
$[3;1,1,1]$, $[2;2,1,1]$, and $[3;2,1,1]$. Enumerations of simple games with up to $4$ voters, together with their corresponding Shapley-Shubik and Banzhaf indices, can e.g.\ be found 
in \cite{straffin1983power}. We remark that there are $28$ different simple games with $4$ voters, where $25$ are weighted. The three non-weighted simple games are characterized by their 
set of minimal winning coalitions $\big\{\{1,2\},\{3,4\}\big\}$, $\big\{\{1,2\},\{1,4\},\{3,4\}\big\}$, and $\big\{\{1,2\},\{1,4\},\{2,3\},\{3,4\}\big\}$, respectively.

In the introduction we have noticed that $[2;1,1,1,1]$ as well as $[3;1,1,1,1]$ yield the power distribution $(0.25,0.25,0.25,0.25)$ both for the 
Shapley-Shubik and the Banzhaf indices, i.e., the power distributions of different games can coincide. In Table~\ref{table_num_power_distributions_weighted} we state the number 
of different power distributions for the Shapley-Shubik and the Banzhaf indices that are attained by weighted games with $n\le 8$ voters. The corresponding numbers for complete 
simple games are listed in Table~\ref{table_num_power_distributions_complete}.  

\begin{table}[htp!]
  \begin{center}
    \begin{tabular}{lrrrrrr}
      \hline
      $n$ & 3 & 4 & 5 & 6 & 7 & 8 \\
      \hline
      $\#\operatorname{SSI}$ & 4 & 11 & 53 & 536 & 14188 & 1364907 \\    
      $\#\operatorname{PBI}$ & 4 & 12 & 57 & 555 & 14720 & 1366032 \\ 
      \hline
    \end{tabular}
    \caption{Number of different vectors $\operatorname{SSI}(v)$ and $\operatorname{PBI}(v)$ for weighted games $v$ with $n$ voters.}
    \label{table_num_power_distributions_weighted} 
  \end{center}
\end{table}

\begin{table}[htp!]
  \begin{center}
    \begin{tabular}{lrrrrrr}
      \hline
      $n$ & 3 & 4 & 5 & 6 & 7 & 8 \\
      \hline
      $\#\operatorname{SSI}$ & 4 & 11 & 53 & 536 & 17973 & 6314952 \\    
      $\#\operatorname{PBI}$ & 4 & 12 & 57 & 555 & 18600 & 4616157 \\ 
      \hline
    \end{tabular}
    \caption{Number of different vectors $\operatorname{SSI}(v)$ and $\operatorname{PBI}(v)$ for complete simple games $v$ with $n$ voters.}
    \label{table_num_power_distributions_complete}
  \end{center}
\end{table}      

We observe that the counts coincide for $n\le 6$, which is no surprise for $n\le 5$, since every complete simple game consisting of at most $5$ voters 
is weighted. However, for $n=6$ voters there exist $1171-1111=60$ complete simple games that are not weighted. Nevertheless, the power distributions according 
to the Shapley-Shubik index or the Banzhaf index of these 60 non-weighted complete simple games are also exactly attained by weighted games, respectively. 

As an example of a non-weighted complete simple game we consider two sets of voters $A=\{1,2\}$ and $B=\{3,4,5,6\}$ such that a coalition $S$ is winning 
iff $|S|\ge 3$ and $|S\cap A|\ge 1$. A representation is given by $[3;1,1,1,1,1,1]\wedge [1;1,1,0,0,0,0]$. 

If we are only interested in the resulting power distribution, then including complete non-weighted games comes with no benefit for $n=6$ voters. For $n\in\{7,8\}$ we 
do not have such a strong result. Here the number of attained power distributions for complete simple games is significantly larger. This goes in line with the fact 
that there are $44\,313-29\,373=14\,940$ and $16\,175\,188-2\,730\,164=13\,445\,024$ non-weighted complete simple games for $n=7$ and $n=8$ voters, respectively. There we can only 
give a worst-case bound for the minimum distance between the power distribution of a complete simple game and a weighted game. To this end, we denote the 
set of weighted games with $n$ voters by $\mathcal{WG}(n)$ and the set of complete simple games with $n$ voters by $\mathcal{CG}(n)$. Moreover, let
$$
  \omega_{a}^p(n):=\max\left\{\min\left\{ d_a(p(c),p(v)) \,:\,v\in\mathcal{WG}(n) \right\}\,:\, c\in \mathcal{CG}(n)\right\},
$$
where $a\in\{1,\infty\}$ and $p\in\{\operatorname{SSI},\operatorname{PBI}\}$, be the worst-case distance between the power distribution $p(c)$ of a 
complete simple game $c$ and the power distribution $p(v)$ of its best approximation by a weighted game $v$.

\begin{proposition}
  \label{prop_exact}
  \begin{eqnarray*}
    \omega_1^{\operatorname{SSI}}(7) \,=\,0.0666667 && \omega_1^{\operatorname{SSI}}(8) \,=\, 0.0666667\\
    \omega_\infty^{\operatorname{SSI}}(7) \,=\,0.0166667 && \omega_\infty^{\operatorname{SSI}}(8) \,=\, 0.0154762\\
    \omega_1^{\operatorname{PBI}}(7) \,=\,0.0599700 && \omega_1^{\operatorname{PBI}}(8) \,=\,0.0567084 \\
    \omega_\infty^{\operatorname{PBI}}(7) \,=\,0.0173913 && \omega_\infty^{\operatorname{PBI}}(8) \,=\,0.0139124 \\
  \end{eqnarray*}
\end{proposition}
\begin{proof}
  The proof is obtained by a computer enumeration. First, we loop over all elements $v$ in $\mathcal{WG}(n)$ and store the corresponding power distributions $p(v)$ 
  in a $k$-d-tree (a data structure for storing multi-dimensional geometrical data). Afterwords, we loop over all elements $c$ in $\mathcal{CG}(n)$, compute $p(c)$, and 
  perform a nearest neighbor search within the previously computed $k$-d-tree. Let $v$ denote the nearest neighbor that minimizes $d_a^p(p(v),p(c))$. Eventually 
  update the worst-case distance with $d_a^p(p(v),p(c))$.   
\end{proof}

As an example we state that the complete simple game attaining $\omega_\infty^{\operatorname{PBI}}(7)=0.0173913$ is uniquely characterized by the subsets   
$\{3,4,5,6,7\}$, $\{2,3,5,6\}$, and $\{1,3,7\}$. For $n=8$ the extremal complete simple games all contain a unique null voter. We remark that the same enumeration 
is computationally infeasible for $n=9$ voters since the numbers $\#\mathcal{WG}(9)=993\,061\,482$ and $\#\mathcal{CG}(9)=284\,432\,730\,174$ are quite large. (See e.g.\ \cite{kartak2015minimal} 
and \cite{freixas2010weighted} for the details.) So, for $n\ge 9$ we can only state lower bounds for $\omega_a^p(n)$:
\begin{proposition}
  \label{prop_lower}
  \begin{eqnarray*}
    \omega_1^{\operatorname{SSI}}(9)      \,\ge\,0.0634922 & \omega_1^{\operatorname{SSI}}(10)      \,\ge\, 0.0634922 & \omega_1^{\operatorname{SSI}}(11)      \,\ge\, 0.0591627 \\
    \omega_\infty^{\operatorname{SSI}}(9) \,\ge\,0.0130953 & \omega_\infty^{\operatorname{SSI}}(10) \,\ge\, 0.0123016 & \omega_\infty^{\operatorname{SSI}}(11) \,\ge\, 0.0109308\\
    \omega_1^{\operatorname{PBI}}(9)      \,\ge\,0.0562 & \omega_1^{\operatorname{PBI}}(10) \,\ge\, 0.0552 & \omega_1^{\operatorname{PBI}}(11) \,\ge\, 0.0552 \\
    \omega_\infty^{\operatorname{PBI}}(9) \,\ge\,0.0110 & \omega_\infty^{\operatorname{PBI}}(10) \,\ge\, 0.0106 & \omega_\infty^{\operatorname{PBI}}(11) \,\ge\, 0.0100 \\
  \end{eqnarray*}
\end{proposition}
\begin{proof}
  Let $a\in\{1,\infty\}$ and $p\in\{\operatorname{SSI},\operatorname{PBI}\}$. In \cite{kurz2012inverse} the inverse power index problem for the Shapley-Shubik index 
  with respect to the Manhattan distance $d_1(\cdot,\cdot)$ and the Chebyshev distance $d_\infty(\cdot,\cdot)$ within the class of weighted, complete simple, or simple games  
  was formulated as an integer linear programming (ILP) problem, which can be solved exactly even for $n>9$, where the number of weighted games is unknown. For the Banzhaf index 
  the problem whether a solution  
  of the inverse power index problem with distance at most $\delta$ exists can be formulated as an ILP. Using the bisection method for $\delta$ the problem 
  can be solved exactly by a sequence of ILPs, see \cite[Appendix A]{kurz2014heuristic} for the details. Thus, given a complete simple game $c$ with $n$ voters we can compute 
  the corresponding power distribution $p(c)$ and exactly solve the inverse power index problem within $\mathcal{WG}(n)$. If $v$ is a weighted game 
  that minimizes $d_a(p(c),p(v))$, then $d_a(p(c),p(v))$ is a lower bound for $\omega_a^p(n)$. As heuristic candidates for the complete simple game $c$ we have used the 
  extremal ones of Proposition~\ref{prop_exact} and added a suitable number of null voters.  
\end{proof}
We remark that we have also tried to use some randomly chosen complete simple games for $c$ in Proposition~\ref{prop_lower}. However, the resulting lower bounds 
for $\omega_a^p(n)$ are rather small. As an example, the value $\omega_1^{\operatorname{SSI}}(7)=0.0666667$ is attained by the complete simple game $c$ characterized  
by the subsets $\{4,5,6,7\}$, $\{2,4\}$, and $\{1\}$. If we add a null voter, the Shapley-Shubik index is given by $$(0.5024,0.1857,0.1024,0.1024,0.03571,0.03571,0.03571,0)$$ 
with best possible approximation $[84;38,27,19,16,9,9,3,0]$, which also shows $\omega_1^{\operatorname{SSI}}(8)\ge 0.0666667$.

For the voting system $c$ of the European Council according to the Lisbon Treaty we cannot solve the inverse power index problem exactly. However, for all 
$a\in \{1,\infty\}$ and all $p\in\{\operatorname{SSI},\operatorname{PBI}\}$ we can find a weighted game $v$ with $d_a(p(c),p(v))<10^{-5}$, which goes in line 
with the computational experiments in \cite{kurz2014heuristic}.   

\vspace*{-3mm}

\section{Conclusion}
\label{sec_conclusion}
\noindent
Does it pay off to use complete simple games instead of weighted games as binary voting systems? If only the resulting power distributions for the Shapley-Shubik or the Banzhaf 
index are relevant, then the answer is probably no. Whether the worst-case deviations stated in Proposition~\ref{prop_exact} can be regarded as negligible might depend on the 
application. For $n>8$ voters our computational experiments suggest that the worst-case deviations might even go down with an increasing number of voters. Proving this claim 
rigorously might be a hard technical challenge. 
  
We have chosen complete simple games as a reasonable superset of weighted games since the underlying ordering of the players can be assumed in many applications. Another reason is 
that the class of simple games is really large\footnote{There are at least $2^{\left(\sqrt{\frac{2}{3}\pi}\cdot 2^n\right)/\left(n\sqrt{n}\right)}$ complete simple games, see \cite{peled1985polynomial}, less than $2^{2^n}$ simple games, and at most $2^{n^2-n+1}$ weighted games, see \cite{zunic2004encoding}.} and realizes a lot of power distributions. E.g., the parameterized target power distribution 
$\beta(n)=\tfrac{1}{2n-1}\cdot (2,\dots,2,1)\in\mathbb{R}^n$ has been studied by \cite{kurz2014heuristic}. For $6\le n\le 18$ there exists a simple game $v_n$ such that 
$\operatorname{SSI}(v_n)=\beta(n)$, while the best approximation within $\mathcal{WG}(n)$ seems to have a deviation of order $\Theta(\tfrac{1}{n})$. At the very least our values 
for $\omega_a^p(n)$ give a lower bound for the corresponding situation where we enlarge the possible target power distributions to those of simple games.


\end{document}